\documentclass[letterpaper, 10 pt, conference]{ieeeconf}  

\IEEEoverridecommandlockouts                              
                                                         
\usepackage{graphicx}
\usepackage{graphics} 
\usepackage{subcaption}
\usepackage[noadjust]{cite}

\usepackage{amsmath} 
\usepackage{amssymb}

\newtheorem{defi}{Definition}[section]
\newtheorem{ass}[defi]{Assumption}

\newtheorem{thm}[defi]{Theorem}

\DeclareMathOperator*{\argmin}{arg\,min}

\title{\LARGE \bf
Greedy Synthesis of Event- and Self-Triggered Controls\\ with Control Lyapunov-Barrier Function
}

\author{Masako Kishida
\thanks{*This work was supported by JST, PRESTO Grant Number JPMJPR22C3, Japan.}
\thanks{Masako Kishida is with the National Institute of Informatics,
        Tokyo 101-8430, Japan
        {\tt\small kishida@nii.ac.jp}}
}

\begin{document}

\maketitle
\thispagestyle{empty}
\pagestyle{empty}

\begin{abstract}
This paper addresses the co-design problem of control inputs and execution decisions for event- and self-triggered controls subject to constraints given by the control Lyapunov function and control barrier function. The proposed approach computes the control input in a way that allows for longer inter-execution intervals, which distinguishes it from many existing event- and self-triggered controllers or control Lyapunov-barrier function controllers. The proposed approach guarantees lower bounds on the minimum inter-execution times. The effectiveness of the proposed approach is demonstrated and compared with existing approaches using a numerical example.
\end{abstract}

\section{INTRODUCTION}

Control Barrier Function (CBF) \cite{WieA07} is a function used to design control inputs to satisfy safety requirements. As the use of automatic control increases in safety-critical systems, CBF is attracting much attention recently.
Applications of CBF can be found such as in robotics \cite{AgrS17} and adaptive cruise control \cite{AmeGT14}.
Furthermore, CBF has been combined with signal temporal logic \cite{LinD19}, model predictive control \cite{ZenZS21} and 
extended to assuring risk-sensitive safety \cite{SinAA23}, making it more versatile for a wider range of applications.
On the other hand, Control Lyapunov Function (CLF) \cite{KhaR88,Kha02}, an extension of the Lyapunov function, has been widely used to design stabilizing controllers for many different problems (see e.g., \cite{KrsK95,FreP96,PriND99,OgrEH01}).
After an integration of CLF and CBF was proposed in \cite{RomJ14,ZakJ16}, it has been shown that CLF and CBF can be combined to form a quadratic program (QP) for computing control inputs that ensure safety while aiming at stability in \cite{AmeXG17, AmeCM19}. Since then, the CLF-CBF QP approach has been used to solve various problems, such as autonomous surface vehicles \cite{BasTP20} and safe stabilization \cite{MesC23}.

With the increasing prevalence of networked control systems, where the communication bandwidth is shared with other tasks or batteries are used in the system elements, it has become crucial to design systems that use communication bandwidth and energy efficiently.
Event- and self-triggered control approaches are effective solutions to minimize unnecessary communication and energy consumption for such systems \cite{PosT11, HeeJT12, MazT08}.

Several approaches have been proposed for safety-critical systems using event- or self-triggered strategies.
An event-triggered control based on input-to-state safe barrier functions by using a state feedback control law $u = k(x)$ was proposed in \cite{TayOC21}. The approach is to bound the difference between the current state and the state used to compute the control input, i.e., the state at the previous execution time instance, to guarantee that the value of the barrier function monotonically decreases.
More recently,  \cite{XiaBC22} proposed an approach of event-triggered control for multi-agent systems with unknown dynamics. 
It synthesizes an event-triggered control with an adaptive affine dynamics that are updated based on the error states to estimate the real system state. A combination of self-triggered control and CLF-CBF QP was proposed in \cite{YanBT19}; however, it does not seem to guarantee the existence of the control input. 
Moreover, it possibly results in continuous control updates if the optimal control input is achieved at the boundary of the QP.

The main contribution of this paper is to introduce approaches to co-designing control inputs and execution time instances for event- and self- triggered controls, with the aim of meeting the constraints given by the control Lyapunov function and control barrier function while reducing the number of executions. This is achieved by computing the control input so as to obtain long inter-execution intervals in a greedy manner. The approach is different from many existing event- and self-triggered controllers that rely on constant feedback laws \cite{HeeJT12} or control Lyapunov-barrier function controllers that account for only the cost of control inputs. It is also shown that the proposed approach does not exhibit Zeno behavior and that the optimization parameters appear in a lower bound on the minimum inter-execution time.

The rest of the paper is organized as follows. After introducing the notation, system model, and basics of control Lyapunov-barrier function in Section \ref{sec:pre}, Section \ref{sec:et} presents the proposed event-triggered control approach. Section \ref{sec:st} discusses the extension to the self-triggered control approach. The performances of those proposed controllers are illustrated and compared with existing controllers in Section \ref{sec:ex}, which is followed by conclusions in Section \ref{sec:conc}.

\section{PRELIMINARIES} \label{sec:pre}

\subsection{Notation}
The sets of real numbers, real vectors of length $n$, and real matrices
of size $n \times m$ are denoted by $\mathbb{R}$, $\mathbb{R}^n$,  and $\mathbb{R}^{n\times m}$, respectively. 
The sets of nonnegative numbers and nonnegative integers are denoted by $\mathbb{R}_{\geq0}$ and $\mathbb{N}$, respectively.

$L_fV (x)$ denotes the Lie derivative of $V(x)$ along the vector field $f(x)$, i.e. $L_fV (x) =\frac{ \partial V (x)
}{\partial x }f(x)$. 

A continuous function $\gamma: \mathbb{R}_{\geq 0}\rightarrow \mathbb{R}_{\geq 0}$ is said to belong to class $\mathcal{K}$ if it is strictly increasing and $\gamma(0)=0$.
A continuous function $\alpha: \mathbb{R}_{\geq 0} \rightarrow \mathbb{R}_{\geq 0}$ is said to belong to class $\mathcal{K}_{\infty}$ if it belongs to class $\mathcal{K}$ and $\alpha(r) \rightarrow \infty$ as $r \rightarrow \infty$.

\subsection{System Model}
This paper deals with a nonlinear affine system in
the form of
\begin{align}\label{eq:sys}
\dot{x} = f(x)+g(x)u, \ x(0)= x_0 \in \mathcal{C} 
\end{align}
where $x \in \mathcal{D} \subset \mathbb{R}^n$ and $u \in  \mathbb{R}^m$ denote the state and the control input of the system, respectively, and $\mathcal{C} \subseteq  \mathcal{D} $ is a safe set that is defined later. 
It is assumed that $\mathcal{D}$ is bounded, and the functions $f(x)$ and $g(x)$ are Lipschitz.

\subsection{Control Lyapunov Function}
\begin{defi}[Control Lyapunov Function]\label{defi:CLF}
A positive definite function $V:\mathcal{D} \rightarrow \mathbb{R}_{\geq 0}$ is called a Control Lyapunov Function (CLF) if it satisfies
\begin{align}\label{eq:CLF}
\inf_{u\in \mathcal{U}} L_fV (x) + L_gV (x)u\leq -\gamma (V (x))
\end{align}
where $\gamma$ is a class $\mathcal{K}$ function.
\end{defi}
The existence of a CLF guarantees asymptotic stabilization of the nonlinear control system \eqref{eq:sys} with any Lipschitz
continuous feedback controller $u(x)$ that satisfies \eqref{eq:CLF} for all $x\in \mathcal{D}$ \cite{AmeGS14,AmeCM19}.

\subsection{Control Barrier Function}

Let define the safe set $\mathcal{C}$ as the superlevel set
of a continuously differentiable function $h : \mathcal{D} \subset \mathbb{R}^n \rightarrow \mathbb{R}$
\begin{align}
\mathcal{C} = \{x \in \mathcal{D}: h(x)\geq 0\}.
\end{align}

\begin{defi}[Control Barrier Function]\label{defi:CBF}
A function $h:\mathcal{D} \rightarrow \mathbb{R}$ is called a Control Barrier Function (CBF)
if it satisfies
\begin{align}\label{eq:CBF}
\sup_{u\in \mathcal{U}} L_fh (x) + L_gh (x)u\geq -\alpha(h (x))
\end{align}
for all $x \in \mathcal{D}$,
where $\alpha$ is a class $\mathcal{K}_{\infty}$ function.
\end{defi}
The existence of a CLF guarantees that the control system is safe \cite{AmeXG17}.

\subsection{CLF-CBF-Based QP}
Motivated by the results on CLF and CBF, it is of interest to obtain a controller that satisfies
\begin{align}
&L_fV (x) + L_gV (x)u\leq -\gamma (V (x)), \label{eq:CLBF_constraints_L}\\
&L_fh (x) + L_gh (x)u\geq -\alpha(h (x)) \label{eq:CLBF_constraints_S}
\end{align}
 so that it is a safe stabilizing controller.

To design such a controller, a CLF-CBF QP was constructed in \cite{AmeP13,AmeXG17}:
\begin{align} \label{eq:CLF-CBF-QP}
\begin{aligned}
 \mathbf{v} ^*(x) &= \argmin_{\mathbf{v}= [u, \delta]^\top \in\mathbb{R}^{m+1}} \frac{1}{2}\mathbf{v}^\top Q \mathbf{v} + c^\top\mathbf{v}\\
\text{s.t. }\ & \left[\begin{matrix} 
L_gV(x) & -1\\
-L_gh(x)&0
 \end{matrix}\right]
 \mathbf{v} \leq 
 \left[\begin{matrix} 
 b_{\text{clf}}(x)\\
 b_{\text{cbf}}(x)
 \end{matrix}\right],
\end{aligned}
\end{align}
where 
\begin{align}\begin{aligned}\label{eq:b}
b_{\text{clf}}(x)&=-L_fV(x)-\gamma(V(x)),\\
 b_{\text{cbf}}(x)&=L_fh(x) +\alpha(h(x)),
 \end{aligned} \end{align}
a positive definite matrix $Q\in \mathbb{R}^{m+1\times m+1}$ and $c \in  \mathbb{R}^{m+1}$ are weights and 
$\delta$ is a relaxation variable that ensures the solvability of the QP.
If $\delta$ is forced to be nonpositive, then the existence of a feasible controller will guarantee the monotonic decrease of the Lyapunov function.

\section{EVENT-TRIGGERED CONTROL} \label{sec:et}
This section proposes a greedy event-triggered control with the control Lyapunov-barrier function for the system \eqref{eq:sys}.

\subsection{Event-triggered controller structure}
The primary idea behind event-triggered control is to update the control input only when necessary to achieve a specified performance condition, thereby reducing the frequency of updates. 
In line with the standard structure of an event-triggered controller, we consider the control inputs that are maintained constant between successive event times, i.e., 
\begin{align}\label{eq:control_law}
u(t)&= u_k, \ t \in [t_k, t_{k+1}),
\end{align}
where $u_k$ is the control input computed at time $t_k$, which is the time instance when the control input is re-computed and the actuator signals are updated. The time instance $t_k$ is determined by 
\begin{align}\label{eq:update_time}
\begin{aligned}
&t_0 = 0, \\
&  t_{k+1} = \inf\{t \in\mathbb{R}:  t > t_k \text{ and } \text{ trigger condition is met} \}.
\end{aligned}
\end{align}

\subsection{Greedy control update}
Here, we introduce a greedy approach for computing the control input $u_k$ at trigger time $t_k$.
To implement it into an event-triggered control, we are interested in a control law that maximizes the inter-execution time. 
For this purpose, we seek the control input that brings the state away from the boundaries of the constraints \eqref{eq:CLBF_constraints_L}, \eqref{eq:CLBF_constraints_S}. 
This is achieved by maximizing the slack variables $\rho_1$ and $\rho_2$ in the new constraints:
\begin{align}\begin{aligned} \label{eq:CLBF_constraints_new}
L_gV (x)u+\rho_1&\leq b_{\text{clf}}(x),\\
- L_gh (x)u+\rho_2 &\leq b_{\text{cbf}}(x),
\end{aligned}\end{align}
where $b_{\text{clf}}(x)$ and $b_{\text{cbf}}(x)$ are defined in \eqref{eq:b}.

Based on the constraints \eqref{eq:CLBF_constraints_new}, we propose to modify the CLF-CBF QP in \eqref{eq:CLF-CBF-QP} as follows: 
\begin{align}\label{eq:CLF-CBF-R-QP}
\begin{aligned}
\mathbf{v}^*(x)&= \argmin_{\mathbf{v}= [u, \rho_1, \rho_2]^\top \in\mathbb{R}^{m+2}} \frac{1}{2}\mathbf{v}^\top Q\mathbf{v} + c^\top\mathbf{v}\\
\text{s.t. }\ &
\left[\begin{matrix} 
L_gV(x)& 1& 0\\
-L_gh(x)&0& 1\\
0 & 0 & -1
 \end{matrix}\right]
 \mathbf{v} 
\leq 
 \left[\begin{matrix} 
b_{\text{clf}}(x)\\
 b_{\text{cbf}}(x)\\
-\varepsilon_{\text{cbf}}
 \end{matrix}\right],
\end{aligned}
\end{align}
where a positive semidefinite matrix $Q\in \mathbb{R}^{m+2\times m+2}$ and  $c \in  \mathbb{R}^{m+2}$ are weights,
$\rho_1$ and $\rho_2$ are slack variables, and $\varepsilon_{\text{cbf}}>0$ is a constant (design parameter) that forces the states to be away from the boundary. 
This is still QP and the solvability of the QP is still ensured as long as $\rho_1$ is not constrained. 
Let $ [u^*(x), \rho_1^*(x), \rho_2^*(x)]^\top =\mathbf{v}^*(x)$.

Typically, a small norm control input $u$ is desired to minimize the control effort while large $\rho_1$ and $\rho_2$ are desired to maximize the inter-execution time. 
Hence,  a possible choice for $Q$ and $q$ is
\begin{align}
Q  = \left[\begin{matrix} 
w_1 & 0 & 0\\
0 & 0 & 0\\\
0 & 0 & 0\
 \end{matrix}\right], \ c=  \left[\begin{matrix} 
0 & -w_2& -w_3
 \end{matrix}\right]
\end{align}
where $w_1\in \mathbb{R}^{m\times m}$ is positive definite, and $w_2, w_3 \geq 0$. 

With \eqref{eq:CLF-CBF-R-QP}, the proposed controller implements the control input $u_k$ defined by
\begin{align}\label{eq:sf}
u_k = u^*(x_k) =\left[\begin{matrix} 
1 & 0 &  0
 \end{matrix}\right]^\top \mathbf{v}^*(x_k),
\end{align}
where $x_k = x(t_k)$.

\subsection{Trigger conditions}
For the states to remain in the safe region, the safety constraint \eqref{eq:CLBF_constraints_S} should be always satisfied. 
However, the satisfaction of the stability constraint \eqref{eq:CLBF_constraints_L} cannot be guaranteed together with the satisfaction of \eqref{eq:CLBF_constraints_S} in general. This is the same for the proposed controller. Yet, we do our best to minimize the time in which \eqref{eq:CLBF_constraints_L}  is violated.

Define 
\begin{align}
p(x) &=p_k(x), \ t\in[t_k, t_{k+1}),\\
q(x) &= q_k(x) ,  \ t\in[t_k, t_{k+1})
\end{align}
where 
\begin{align}
p_k(x) &= -L_fV(x) - L_gV(x)u_k,\\
q_k(x) &= L_gh(x)u_k  +b_{\text{cbf}}(x).
\end{align}
Note that the time derivative of the Lyapunov function is $\dot{V}(x)  = -p_k(x) $, thus $p_k(x) \geq 0$ is desired for the stability, while $q_k(x) \geq 0$ is desired for the safety.

We set the trigger condition in \eqref{eq:update_time} as
\begin{enumerate}
\item if $ p_k(x_k) \geq \varepsilon_{\text{clf}} $:
\begin{align}
p_k(x)=0 \text{ or } q_k(x)  = 0 \label{eq:tc1}
\end{align}
\item else:
\begin{align}
 q_k(x)  =0   \text{ or } t = t_k + \tau_{bd} \label{eq:tc2}
\end{align}
\end{enumerate}
where $\varepsilon_{\text{clf}}>0$ and $ \tau_{bd} >0$ are design parameters.

The first case is that, if the time-derivative of the Lyapunov function is sufficiently negative at the time of update $t_k$, then the next update time is when either the safety constraint \eqref{eq:CLBF_constraints_L} is satisfied with equality or the time-derivative of Lyapunov function becomes zero. This guarantees the satisfaction of both safety and stability between $t_k$ and $t_{k+1}$.

The second case is that, if the time-derivative of the Lyapunov function is close to zero or positive at the time of update $t_k$, then we compromise the controller design only focusing on the safety constraint.  The next update time is when the safety constraint \eqref{eq:CLBF_constraints_L} is satisfied with equality or small time $\tau_{bd}$ passes, whichever occurs first. This limits the duration of time during which the stability constraint is violated to $\tau_{bd}$ with the same control input.

\subsection{Lower bound on minimum inter-execution time}\label{sec:lb}
This subsection shows that the events cannot be triggered an infinite number of times in
any finite time period with the proposed controller, i.e., the proposed controller is Zeno-free under mild conditions.

Define the inter-execution time 
\begin{align}
\tau_k = t_{k+1} - t_{k}, \ k \in \mathbb{N}.
\end{align}
It is of interest to show the existence of a lower bound $\tau^*$ such that $\tau_k \geq \tau^*$ for all $k \in \mathbb{N}$.

\begin{ass} \label{ass}
We assume the followings hold on $\mathcal{D}$:
\begin{itemize}
\item $\|f(x)+g(x)u\|$ is bounded above, and
\item There exists a Lipschitz constant $L_{\text{clf}}> 0$ for $p(x)$ 
\item There exists a Lipschitz constant $L_{\text{cbf}}> 0$ for $q(x)$
\end{itemize}
\end{ass}
If the problem is considered in a finite time horizon, the first assumption is sufficient.

\begin{thm} \label{thm:lbound}
Under Assumption \ref{ass},
for the time sequence \eqref{eq:update_time} for the system \eqref{eq:sys} with the event-triggered controller \eqref{eq:control_law}, \eqref{eq:sf} with the trigger condition \eqref{eq:tc1}-\eqref{eq:tc2}, 
there exists $\tau^*>0$ such that $\tau_k \geq \tau^*$ for all $k\in \mathbb{N}$.
\end{thm}
\begin{proof}
First, we consider if $ p_k(x_k) \geq \varepsilon_{\text{clf}}$, then how long it takes to achieve
\begin{align}\label{eq:stab_cond}
p_k(x)= 0 
\end{align}
for the first time after $t_k$. Let this time instance be $\bar{t}_{\text{clf}}$ and the corresponding state be $\bar{x}_{\text{clf}}$.

For this, we first show that  
  \begin{align} 
\|\bar{x}_{\text{clf}}-x_k\| \geq \frac{\varepsilon_{\text{clf}}}{L_{\text{clf}}}. 
\end{align}
Clearly, 
 \begin{align} 
p_k(x_k)- p_k(\bar{x}_{\text{clf}}) \geq \varepsilon_{\text{clf}},
\end{align}
and 
  \begin{align} 
L_{\text{clf}} \|x_k -\bar{x}_{\text{clf}} \| \geq p_k(x_k)- p_k(\bar{x}_{\text{clf}}).
\end{align}
Together, it follows that 
\begin{align}
\|\bar{x}_{\text{clf}}-x_k\| \geq \frac{\varepsilon_{\text{clf}}}{ L_{\text{clf}} }.
\end{align}

Next, we show that the time difference $\bar{t}_{\text{clf}}-t_k$ is lower bound away from zero. 
By the fundamental theorem of calculus, we have
\begin{align}\begin{aligned}
\|\bar{x}_{\text{clf}}-x_k\| &= \left\|\int_{t_k}^{\bar{t}_{\text{clf}}} \dot{x}(\tau)d\tau\right\| \\
&\leq \sup  \|f(x)+g(x)u \| (\bar{t}_{\text{clf}}-t_k).
 \end{aligned} \end{align}
Because $\|f(x)+g(x)u \|$ is bounded above, there exists $M>0$ such that $M= \sup  \|f(x)+g(x)u \|$, then it follows that
 \begin{align}
\bar{t}_{\text{clf}}-t_k \geq \frac{1}{M}\|\bar{x}_{\text{clf}}-x_k\|\geq  \frac{\varepsilon_{\text{clf}}}{ML_{\text{clf}}}.
 \end{align}
Thus, at least the time interval of $ \frac{\varepsilon_{\text{clf}}}{ML_{\text{clf}}}$ takes to achieve \eqref{eq:stab_cond}.

Similarly, we consider
how long it takes to achieve
\begin{align}\label{eq:safe_cond}
q_k(x)= 0 
\end{align}
for the first time after $t_k$. Let this time instance $\bar{t}_{\text{cbf}}$ and the corresponding state $\bar{x}_{\text{cbf}}$.

Using the fact that $q_k(x_k) \geq\rho_2^*(x_k)  \geq \varepsilon_{\text{cbf}}$, 
we can show that
  \begin{align} 
\|\bar{x}_{\text{cbf}}-x_k\| \geq \frac{\varepsilon_{\text{cbf}}}{L_{\text{cbf}}}
\end{align}
 and then
 \begin{align}
\bar{t}_{\text{cbf}}-t_k \geq \frac{1}{M}\|\bar{x}_{\text{cbf}}-x_k\|\geq  \frac{\varepsilon_{\text{cbf}}}{ML_{\text{cbf}}}.
 \end{align}

In summary,   the inter-execution time is lower bounded by 
$\tau^*=\min(\frac{\varepsilon_{\text{clf}}}{ML_{\text{clf}}}, \frac{\varepsilon_{\text{cbf}}}{ML_{\text{cbf}}}, \ \tau_{bd})$,
which is strictly positive.
This completes the proof.
\end{proof}

Here, we see the parameter $\varepsilon_{\text{cbf}}>0$ in \eqref{eq:CLF-CBF-R-QP} can be used as a parameter to tune the inter-execution time.
\subsection{Special cases}
\subsubsection{Feasibility is guaranteed}
Suppose that the feasibility of the constraints \eqref{eq:CLBF_constraints_L} and  \eqref{eq:CLBF_constraints_S} is guaranteed for all $x \in \mathcal{D}$ with some margins, i.e., 
there exist $s_1, s_2 > 0$ such that for all $x \in \mathcal{D}$, there exists an control input $u$ that depends on $x$ that satisfies 
\begin{align}\begin{aligned} 
L_gV (x)u+s_1&\leq b_{\text{clf}}(x),\\
- L_gh (x)u +s_2 &\leq b_{\text{cbf}}(x).
\end{aligned}\end{align}
Then, we may consider not only forcing the relaxation variable $\delta$ to be 0 in \eqref{eq:CLF-CBF-QP}, but securing certain distances from the boundaries of the constraints.
This allows us to add a constraint $\rho_1\geq0$ in \eqref{eq:CLF-CBF-R-QP} and the resulting QP is:
\begin{align}\label{eq:CLF-CBF-R-QP2}
\begin{aligned}
\mathbf{v}^*(x) &= \argmin_{\mathbf{v}= [u, \rho_1, \rho_2]^\top \in\mathbb{R}^{m+2}} \frac{1}{2}\mathbf{v}^\top Q\mathbf{v} + c^\top\mathbf{v}\\
\text{s.t. }\ &
\left[\begin{matrix} 
L_gV(x)& 1& 0\\
-L_gh(x)&0& 1\\
0 & -1 & 0\\
0 & 0 & -1
 \end{matrix}\right]
 \mathbf{v} 
\leq 
 \left[\begin{matrix} 
b_{\text{clf}}(x)\\
 b_{\text{cbf}}(x)\\
-\varepsilon_1 \\
-\varepsilon_2
 \end{matrix}\right],
\end{aligned}
\end{align}
where $ \varepsilon_1\in (0, s_1), \varepsilon_2  \in (0, s_2)$ are constants, a positive semidefinite matrix $Q\in \mathbb{R}^{m+2\times m+2}$, and $c \in  \mathbb{R}^{m+2}$ are the weights. 
In this case, the trigger condition \eqref{eq:tc1}-\eqref{eq:tc2} can be combined and modified as:
\begin{align}
(L_gV(x)u  - b_{\text{clf}} (x) )(L_gh(x)u  +b_{\text{cbf}}(x) ) = 0\label{eq:tc3}
\end{align}
to guarantee both stability and safety.
\subsubsection{Control input is not penalized}
In addition that the feasibility is guaranteed, if the control input $u$ is not penalized as in other event-triggered controls, then, the problem simplifies to solving a linear program:
\begin{align}\label{eq:CLF-CBF-R-QP2}
\begin{aligned}
u^*(x) &= \argmin_{u \in\mathbb{R}^{m}} \left(-w_2L_gV(x)-w_3L_gh(x)\right)u \\
\text{s.t. }\ &
\left[\begin{matrix} 
L_gV(x)\\
-L_gh(x)
 \end{matrix}\right]
u
\leq 
 \left[\begin{matrix} 
b_{\text{clf}}(x)-\varepsilon_1\\
 b_{\text{cbf}}(x)-\varepsilon_2
 \end{matrix}\right],
\end{aligned}
\end{align}
where, again, $ \varepsilon_1\in (0, s_1), \varepsilon_2  \in (0, s_2)$ are constants, 
using some weights $w_2, w_3 \geq 0$.

\section{Self-triggered control} \label{sec:st}
This section develops a greedy self-triggered control based on the results in Section \ref{sec:et}.

\subsection{Self-triggered controller structure}
As in Section \ref{sec:et}, let $t_k$ be the triggering time instance.
Self-triggered control computes the control input $u_k$ as well as the next time instance $t_{k+1}$ at execution time $t_k$. 
Similar to the event-triggered control \eqref{eq:control_law}, the control input $u_k$ remains constant in between $t_k$ and $t_{k+1}$. 
Unlike the event-triggered condition, however, 
no sampling or computation is required between $t_k$ and $t_{k+1}$.

As for the event-triggered control, the controller implements the control input $u_k$ in \eqref{eq:sf} by solving 
 \eqref{eq:CLF-CBF-R-QP}.

The next execution time $t_{k+1}$ is computed based on the measured state $x(t_k) = x_k$ at $t_k$:
\begin{align}\label{eq:update_time2}
\begin{aligned}
&t_0 = 0, \\
&  t_{k+1} = t_k + \Gamma(x_k),
\end{aligned}
\end{align}
where the map $\Gamma : \mathbb{R}^n \rightarrow \mathbb{R}_{\geq 0}$ determines the
triggering time $t_{k+1}$ as a function of the state $x_k$ at the time $t_k$. Thus, the inter-execution time is given by $\Gamma (x)$, i.e.,  $\tau_k =  \Gamma(x_k)$.

In the following, we present approaches to how to design the map $\Gamma(x)$.
\subsection{Computing the next execution time instance}

In Section \ref{sec:et}, the trigger condition \eqref{eq:tc1}-\eqref{eq:tc2} is designed to satisfy 
the safety constraint \eqref{eq:CLBF_constraints_S} and minimize the violation of the stability constraint \eqref{eq:CLBF_constraints_L}. Here, again, we design the map $\Gamma$ that satisfies
the safety constraint \eqref{eq:CLBF_constraints_S} all the time, while allowing the violation of the stability constraint:
\begin{align}\begin{aligned}
& \Gamma(x_k) \leq \sup\{ t >t_k:\\
& \
\text{  if } p_k(x_k) \geq  \varepsilon_{\text{clf}}: \\
& \qquad p_k(x(\tau))\geq 0 \text{ and } q_k(x(\tau))  \geq 0  \text{ for all } \tau \leq t\\
&\
\text{ else: }\\
&\qquad p_k(x(\tau))\geq 0  \text{ for all } \tau \leq t \text{ and } t \leq t_k + \tau_{bd} 
\}\\
& \ -t_k. \label{eq:gamma_bd}
\end{aligned}\end{align}

Ideally, we would like to find $t$ that achieves an equality in \eqref{eq:gamma_bd}.
However, it is difficult in general for nonlinear systems, thus, we aim at finding a lower bound by revising the approach in Section \ref{sec:lb}.
 
Again, we assume that Assumption \ref{ass} is satisfied. Then, we may use the following map:
\begin{thm} \label{thm:st_bound}
Under Assumption \ref{ass}, 
\begin{align} \label{eq:st_app}
\Gamma(x_k) = 
\begin{cases}
\min\left( \frac{\varepsilon_{\text{clf}} }{L_{\text{clf}}M_k},  \frac{\varepsilon_{\text{cbf}} }{L_{\text{cbf}} M_k}\right), \text{ if } p_k(x_k) \geq  \varepsilon_{\text{clf}}\\ 
\min\left( \frac{\varepsilon_{\text{cbf}} }{L_{\text{cbf}} M_k} ,\tau_{bd}  \right), \text{ otherwise },\\
\end{cases}
 \end{align}
 where 
 \begin{align} \label{eq:M}
 M_k = \sup_{t \in\left[t_k, t_k + \Gamma(x_k) \right)} \|f(x)+g(x)u_k\|
 \end{align}
  satisfies \eqref{eq:gamma_bd}.
\end{thm}
\begin{proof}
 Here, again note that after solving \eqref{eq:CLF-CBF-R-QP}, it is guaranteed that 
  \begin{align}
   q_k(x_k) \geq \rho_2^*(x_k) \geq \varepsilon_{\text{cbf}}.
  \end{align}  
  
For the first case of $p_k(x_k) \geq \varepsilon_{\text{clf}}$, we show
\begin{align}
\Gamma(x_k) =\min\left( \frac{\varepsilon_{\text{clf}} }{L_{\text{clf}}M_k},  \frac{\varepsilon_{\text{cbf}} }{L_{\text{cbf}} M_k} \right), \label{eq:st_app1}
 \end{align}
  satisfies \eqref{eq:gamma_bd}.
By the fundamental theorem of calculus, the equation \eqref{eq:st_app1} implies that for $ t \in[t_k, t_k +\Gamma(x_k) )$,
\begin{align}\begin{aligned}
\|x -x_k\|& \leq \int_{t_k}^{t} \|f(x) +g(x)u_k\| dt \\
&\leq  (t-t_k ) M_k \\
&\leq  (t-t_k ) \frac{ \min\left( \frac{\varepsilon_{\text{clf}} }{L_{\text{clf}}},  \frac{\varepsilon_{\text{cbf}} }{L_{\text{cbf}}} \right)}{\Gamma(x_k) }\\
&\leq   \min\left( \frac{\varepsilon_{\text{clf}} }{L_{\text{clf}}},  \frac{\varepsilon_{\text{cbf}} }{L_{\text{cbf}}} \right)\\
\Leftrightarrow \
&\begin{cases}
\varepsilon_{\text{clf}}  -L_{\text{clf}}\|x -x_k\|\geq 0,  \\ 
\varepsilon_{\text{cbf}}  -L_{\text{cbf}}\|x -x_k\|\geq 0.
\end{cases}
 \end{aligned} \end{align}
On the other hand, 
\begin{align}\begin{aligned}
\varepsilon_{\text{clf}}- p_k(x) \leq p_k(x_k) - p_k(x) \leq L_{\text{clf}} \|x -x_k\|, \\
\varepsilon_{\text{cbf}}- q_k(x) \leq  q_k(x_k) - q_k(x) \leq  L_{\text{cbf}} \|x -x_k\|.
  \end{aligned}  \end{align}
Hence, it follows that
  \begin{align}
 p_k(x) \geq 0, \  q_k(x) \geq 0.
  \end{align}
  
The second case is clear from the above discussions.
Together, it completes the proof.
\end{proof}

The difference between this map $\Sigma$ and the lower bound $\tau^*$ in the event-triggered control is only the upper bound on the norm of $f(x)+g(x)u$. Although a uniform $M$ may be used to design a constant $\Gamma=\tau^*$, such a law will shorten the inter-execution time because $M$ is likely to be much larger than $M_k$. 
However, one difficulty of implementing this approach is actually in the computation of $M_k$ in \eqref{eq:M} where $M_k$ appears in both sides of the equation. To compute exact $M_k$, it is required to compute the evolution of \eqref{eq:sys} starting at $t_k$ by gradually increasing the time duration. 
In the actual implementation, this might not be a big problem because implementation is done digitally, which we discuss next.
\subsection{Digital implementation}
Similarly to \cite{HeeJT12}, we now consider the following discrete-time versions of $p_k(x)$ and $q_k(x)$ based on a sampling time $\Delta >0$, which are defined by
  \begin{align}  \begin{aligned}
  \bar{p}_k(n) &=  p_k(x(t_k+n\Delta)),\\
\bar{q}_k(n) &=  q_k(x(t_k+n\Delta)).
  \end{aligned}  \end{align}

Let $\tau_{\min}$ and $\tau_{\max}$ be design parameters and let $N_{\min} = \lfloor \tau_{\min}/\Delta\rfloor$, $N_{\max} = \lfloor \tau_{\max}/\Delta\rfloor$. 

Then, by choosing $\tau_{bd} = \Delta$, the map can be simplified as
\begin{align}
\Gamma(x_k) =  \max \{ \tau_{\min}, n(x_k)\Delta\}, \label{eq:st_app3}
 \end{align}
 where 
 \begin{align} \label{eq:M} \begin{aligned}
n(x) = \max_{n \in \mathbb{N} } \{ n \leq N_{\max} : \bar{p}_k(m)& \geq 0 \text{ and } \bar{q}_k(m) \geq 0, \\
&\qquad m = 1, \cdots, n\}. 
 \end{aligned} \end{align}
This satisfies 
  \begin{align}
\bar{p}_k(n) \geq 0 \text{ and } \bar{q}_k(n) \geq 0, \ \forall n \in \left[0, \left\lceil \frac{t_{k+1}-t_k}{\Delta} \right\rceil\right)
  \end{align}
and $n(x) \in [N_{\min}, N_{\max} ]$.

With a smaller choice of $\Delta$, the stability constraint is more strictly forced, i.e., while the stability constraint is violated, the control input is updated every sampling time. Such a choice of $\Delta$ may result in an increase in the number of execution. Of course, it is possible to use different values for $\tau_{bd}$ and $\Delta$ with a slight modification. 

We may simply choose  $\tau_{\min}=0$ and a sufficiently large $\tau_{\max}$.
However,   the value of $\tau_{\max}$ enforces the robustness of the implementation and limits the computational complexity \cite{HeeJT12}.

\section{NUMERICAL EXAMPLE} \label{sec:ex}
This section demonstrates the effectiveness of the proposed approach using an example of a double integrator.

Let $x =  \left[\begin{matrix}
 x_1&x_2
 \end{matrix}\right]^\top$, where $x_1$ is the position and $x_2$ is the velocity. 
 The dynamics of a double integrator is given as follows:
 \begin{align} \label{eq:example} 
 \dot{x} &= \left[\begin{matrix}
 0 & 	1 \\ 0 &0
 \end{matrix}\right] x+ 
 \left[\begin{matrix}
 0\\ 1
 \end{matrix}\right]u.
 \end{align}

The control Lyapunov and control barrier functions are selected as
 \begin{align} \begin{aligned}
 V(x) &= x_1^2+x_1x_2+x_2^2, \\
 h(x) &=  (x_1-0.5)^2+(x_2+0.5)^2-0.3^2.
\end{aligned}\end{align}

Moreover, the $\alpha$ an $\gamma$ functions are chosen to be identity maps and 
the parameters $\varepsilon_{\text{clf}}=\varepsilon_{\text{cbf}}=0.1$, $\tau_{bd}= 0.5$, $\Delta =0.2  $, $\tau_{\min} =0$ and $\tau_{\max}=4$ are selected.
The weights for QP in \eqref{eq:CLF-CBF-R-QP} are selected as
\begin{align}
 Q = \left[\begin{matrix}
1 & 0	&0\\ 0 &0&0\\ 0 &0&0
 \end{matrix}\right], \ c = \left[\begin{matrix}
0& -1	&0 \end{matrix}\right]^\top.
\end{align}

Here, the performances of the following five controllers are compared for the duration of time 15, starting at $x_0 = [1, 1]$.
\begin{itemize}
\item Greedy ET: the controller that solves \eqref{eq:CLF-CBF-R-QP} and implements \eqref{eq:sf} when the trigger condition \eqref{eq:tc1} or \eqref{eq:tc2} is met
\item Greedy ST: the controller that solves \eqref{eq:CLF-CBF-R-QP} and implements \eqref{eq:sf} at time instances determined by \eqref{eq:update_time2} and  \eqref{eq:st_app3}
\item Greedy: the controller that solves \eqref{eq:CLF-CBF-R-QP} and implements \eqref{eq:sf} continuously
\item CLF-CBF QP: CLF-CBF QP controller in \cite{AmeCM19} with $H(x)=2$, $p=1$
\item SF: a standard state-feedback controller with the controller gain $K= \left[-0.5, -1\right]$, which corresponds to the case of weight for the states is $\left[\begin{matrix}0.25 & -0.5\\ -0.5 & 0\end{matrix}\right]$ and weight for the control input is $1$, thus $P= \left[\begin{matrix}1 & 0.5\\ 0.5 & 1\end{matrix}\right]$ in the algebraic Riccati equation. No constraints are considered.
\end{itemize}

\begin{figure}[tbh]\centering
 \includegraphics[width=.98\linewidth, viewport =20 0 520 260, clip]{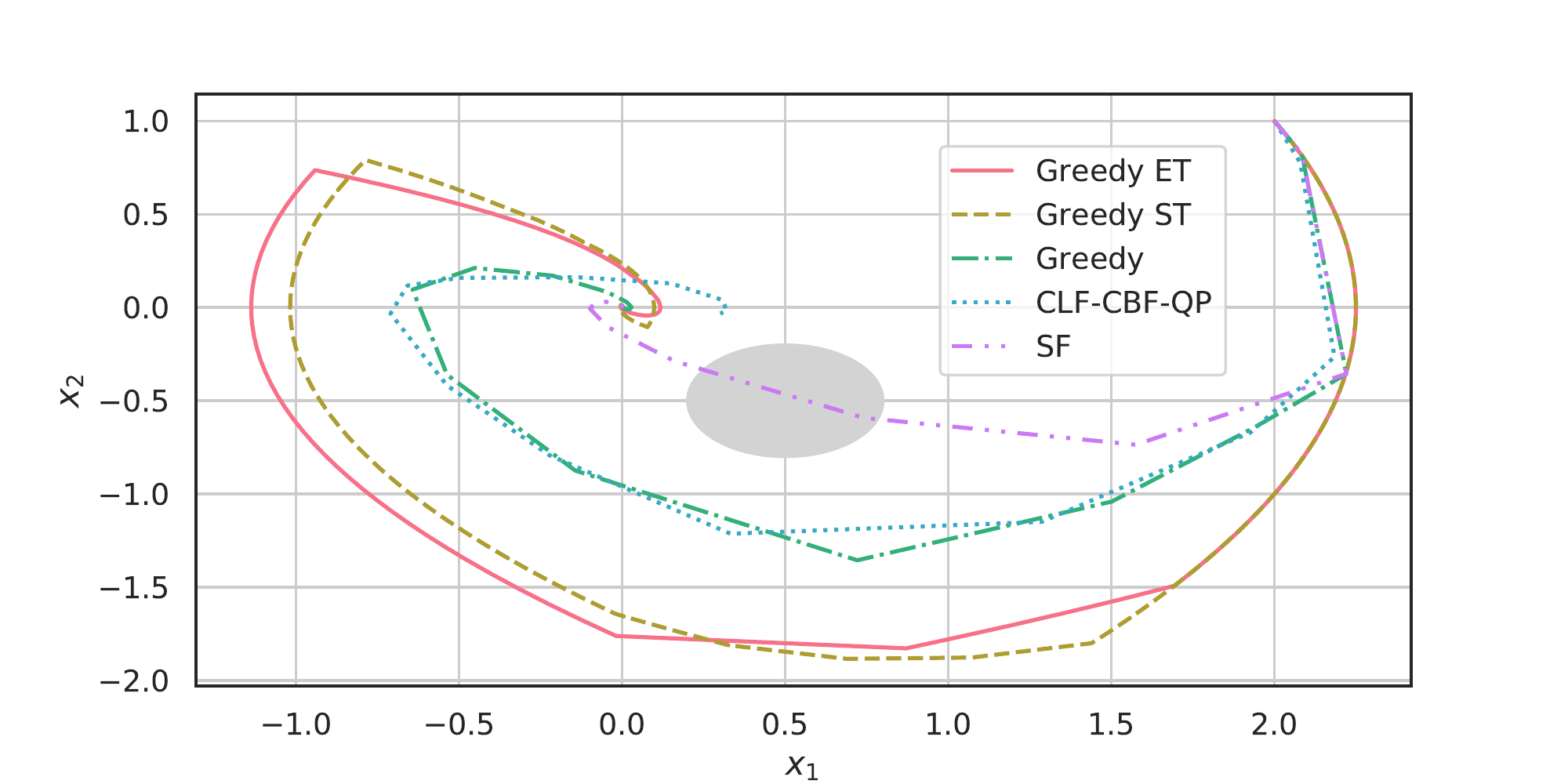}
\caption{Phase portrait: the grey region indicates the unsafe region, $h(x)<0$ } 
\label{fig:pp}\vspace{-.1in}
\end{figure}

\begin{figure}[tbh]\centering
 \includegraphics[width=.98\linewidth, viewport =20 0 520 260, clip]{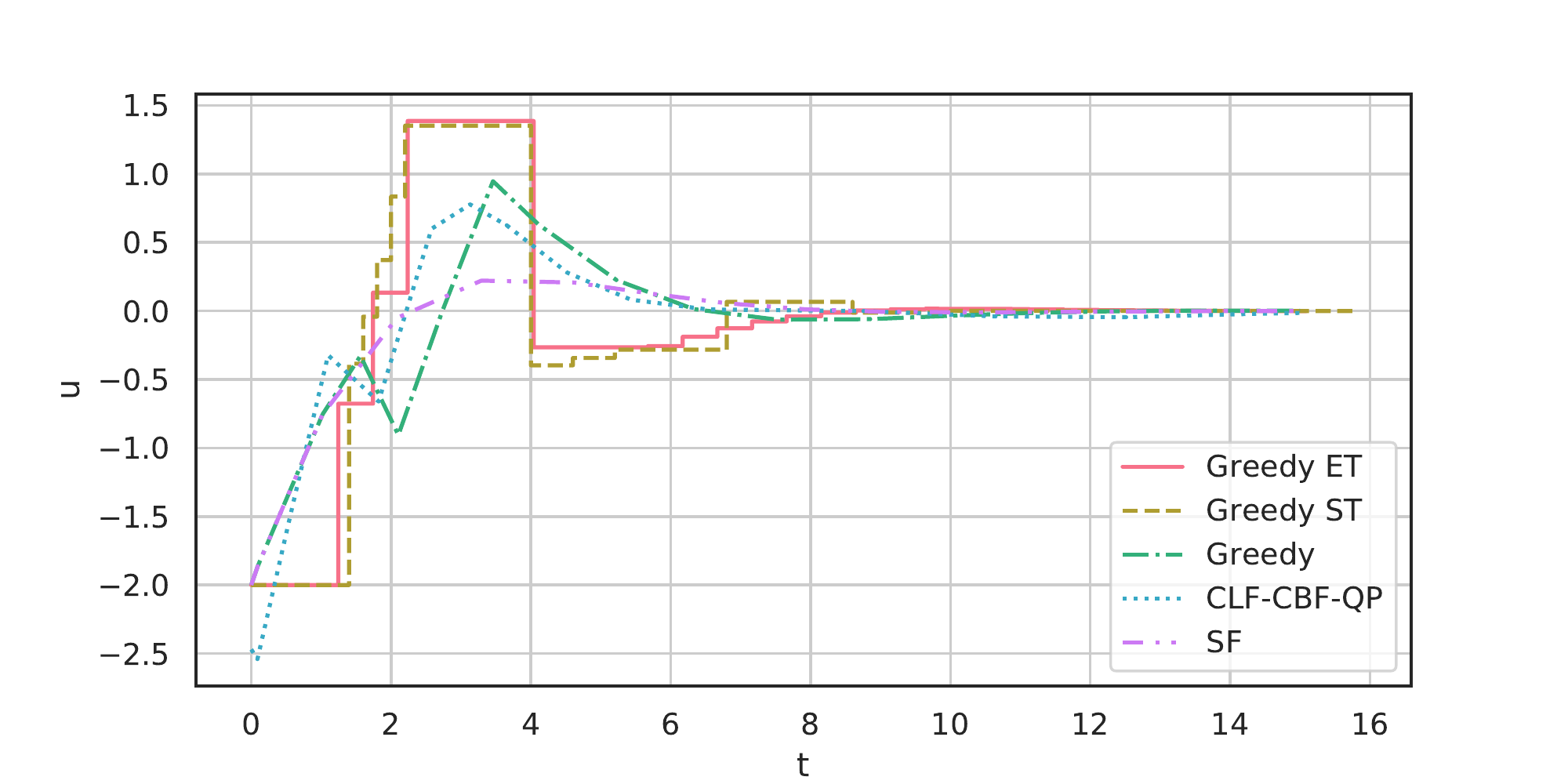}
\caption{Control inputs} 
\label{fig:u}\vspace{-.1in}
\end{figure}

\begin{figure}[tbh]\centering
 \includegraphics[width=.98\linewidth, viewport =20 0 520 260, clip]{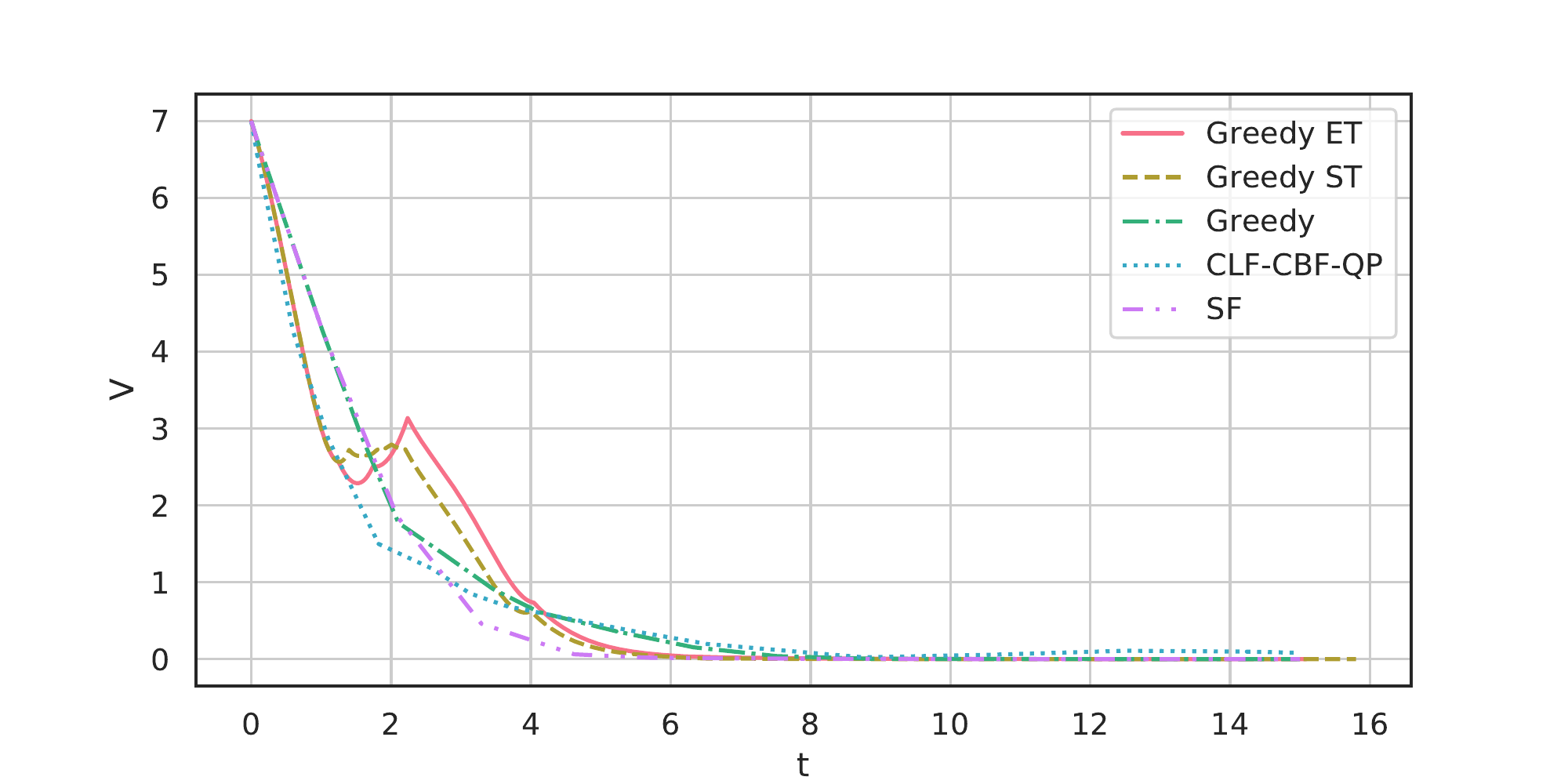}
\caption{Lyapunov function values, $V(x)=x_1^2+x_1x_2+x_2^2$} 
\label{fig:V}\vspace{-.1in}
\end{figure}

\begin{figure}[tbh]
\begin{minipage}[b]{0.95\hsize}
    \centering
 \includegraphics[width=.98\linewidth, viewport =20 0 520 260, clip]{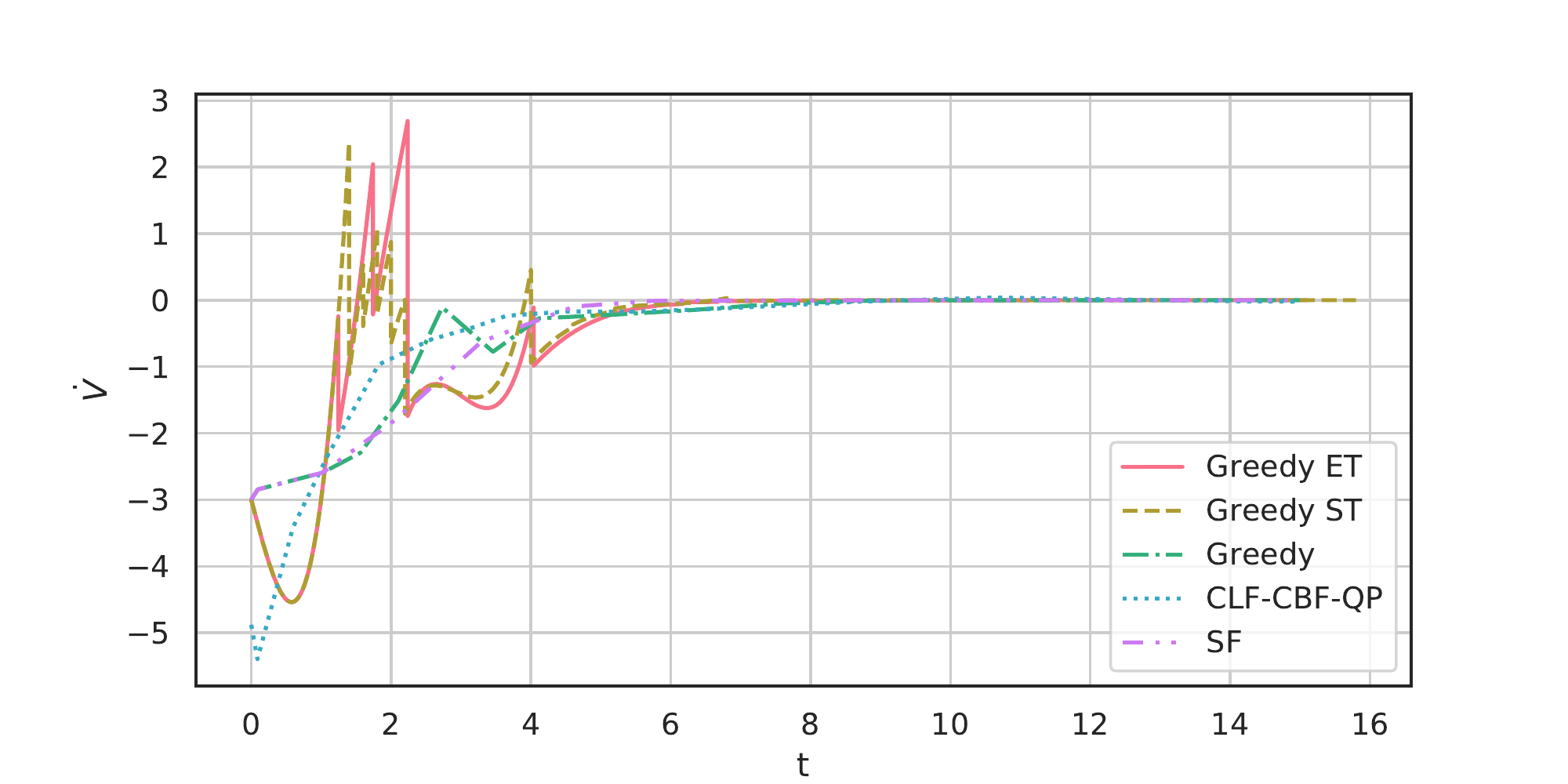}
    \subcaption{Stability constraint used for trigger condition, $L_fV(x)+L_gV(x)u$. The values are desired to be negative. }\label{fig:Vdot}
  \end{minipage}
  \begin{minipage}[b]{0.95\hsize}
    \centering
 \includegraphics[width=.98\linewidth, viewport =20 0 520 260, clip]{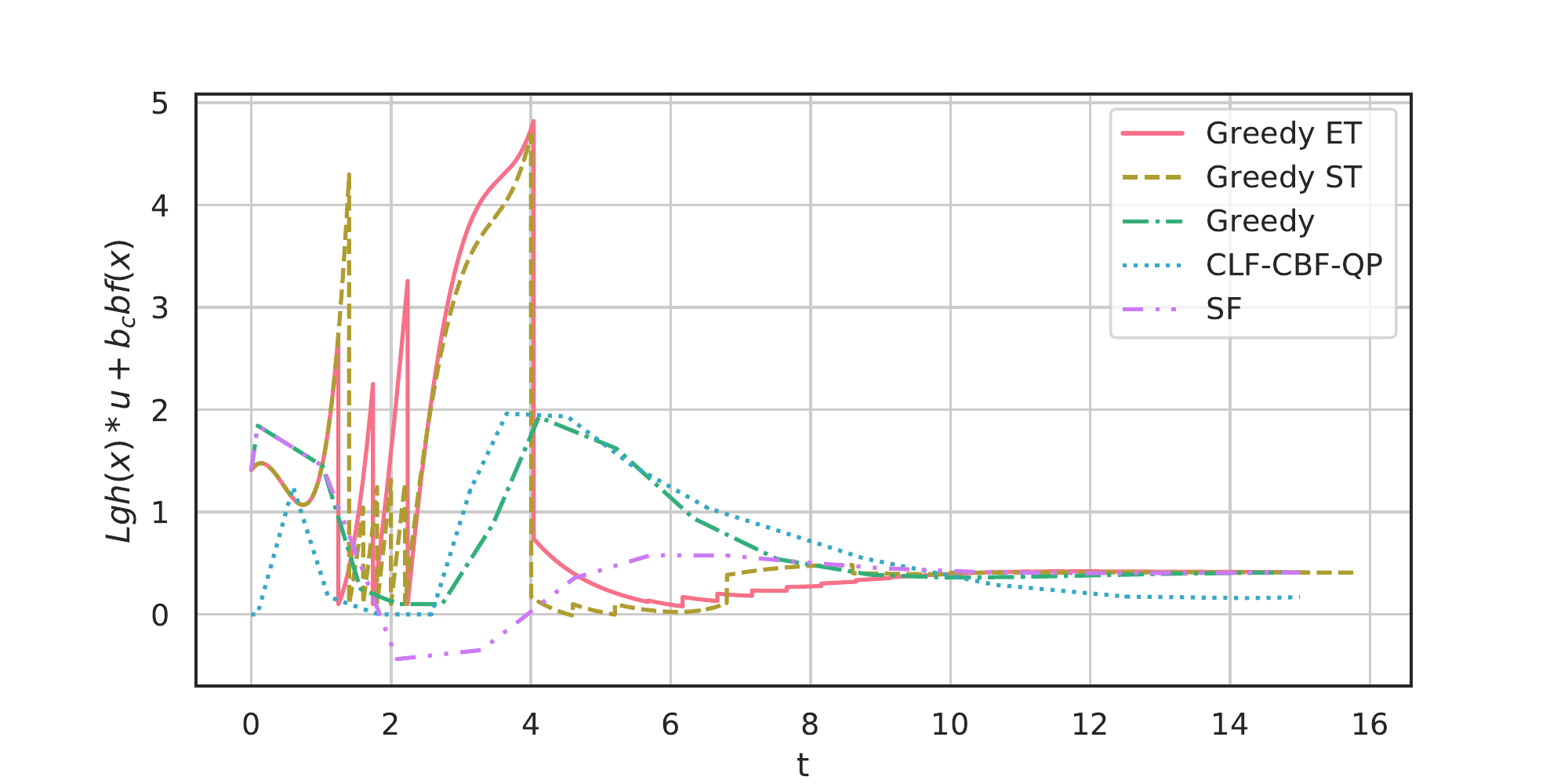}
    \subcaption{Safety constraint used for trigger condition,  $L_gh(x)u + b_{cbf}(x)$. The values are desired to be nonnegative. }\label{fig:hdot}
  \end{minipage}
  \caption{Trigger conditions}\label{fig:trig_cond}
\end{figure}

Figure \ref{fig:pp} shows the phase portraits for the five controllers. It is observed that
the trajectory with SF goes into the unsafe region, which motivates us to use the barrier function to remain in the safe region.
Also, the trajectories of Greedy and CLF-CBF-QP are close to each other, but the state of CLF-CBF-QP is further from the origin at the end of the simulation compared with the other four methods.
Moreover, both trajectories of Greedy ET and Greedy SF take longer paths compared with other methods.

Figure \ref{fig:u} shows the control input trajectories. It can be seen that Greedy ET and Greedy ST do not require frequent control updates.
In fact, the numbers of control updates for Greedy ET and Greedy ST were  24 and 26, respectively. 
With Greedy ST,  a smaller sampling time $\Delta$ forces the derivative of the Lyapunov function to be negative more strictly thus tends to increase the update frequency.

Figure \ref{fig:V} shows the Lyapunov function values. Both  Greedy ET and Greedy ST admit increases of the Lyapunov function values for certain periods. However, those trajectories approach zero quickly, while CLF-CBF-QP is still away from zero at the end of the simulation. Note that the original CLF-CBF-QP controller also allows increases of the Lyapunov function values to guarantee the feasibility of the optimization problem \cite{AmeCM19}.

Figure \ref{fig:trig_cond} shows the trajectories of values used for triggers.
Because Greedy ET and Greedy ST  compromised the stability constraints for the sake of reducing the update frequencies, Figure \ref{fig:Vdot} indicates the values of $L_fV(x)+L_gV(x)u$ go above zero sometimes.
On the other hand, the safety constraints are always satisfied by all the controllers except for the SF that ignored the existence of unsafe region Figure \ref{fig:hdot}. 
In Figure \ref{fig:hdot}, we also observe that the trajectory $L_gh(x)u + b_{cbf}(x)$ of CLF-CBF-QP stays near zero around time 2. 
Thus, we cannot implement an event-triggered strategy with CLF-CBF-QP because the trigger condition is violated or close to violate already at the time of the update.

\section{CONCLUSIONS}\label{sec:conc}
In this paper, we have presented a set of greedy approaches for synthesizing event-triggered and self-triggered controls with the control Lyapunov-Barrier function. Our proposed approach computes each control input to maximize the distance from the safety boundary, which is a departure from existing approaches to the control Lyapunov-Barrier function. By doing so, our approach ensures a positive lower bound on the minimum inter-execution time, while also maintaining the safety of the control system and reducing the frequency of control input updates and/or samplings. This is particularly beneficial in the context of networked control systems, where safety is of paramount concern. The effectiveness of our proposed approach has been illustrated through a numerical example, which highlights its potential for real-world implementation.


\bibliographystyle{IEEEtran}
\bibliography{IEEEabrv,myref}

\end{document}